\documentclass[12pt,reqno]{amsart}
\usepackage{amsmath,amsfonts,amsthm,amssymb}
\newcommand\version{July 26, 2023}
\usepackage{amsxtra}

\newtheorem{theorem}{Theorem}

\newtheorem{corollary}[theorem]{Corollary}

\theoremstyle{definition}

\theoremstyle{remark}

\renewcommand{\epsilon}{\varepsilon}

\renewcommand{\phi}{\varphi}
\newcommand{\R}{\mathbb{R}}

\DeclareMathOperator{\Tr}{Tr}

\begin{document}

\title[Simple approach to LT type inequalities  --- RS \& JPS,  \version]{A simple approach to Lieb--Thirring type inequalities}

\author{Robert Seiringer}
\address{Robert Seiringer, IST Austria, Am Campus 1, 3400 Klosterneuburg, Austria}
\email{rseiring@ist.ac.at}

\author{Jan Philip Solovej}
\address{Jan Philip Solovej, Department of Mathematics, University of Copenhagen, Universitetsparken 5, DK-2100 Copenhagen \O, Denmark}
\email{solovej@math.ku.dk}

\thanks{\copyright\, 2023 by the authors. This paper may be  
reproduced, in
its entirety, for non-commercial purposes.}

\begin{abstract} 
In \cite{Nam} Nam proved a Lieb--Thirring Inequality for the kinetic energy of a fermionic quantum system, with almost optimal (semi-classical) constant and a gradient correction term. We present a stronger version of this inequality, with a much simplified proof. As a corollary we obtain a simple proof of the original Lieb--Thirring inequality.
\end{abstract}

\date{\version}

\maketitle


Let $\gamma$ be a positive trace-class operator on $L^2(\R^d)$ with density (i.e., diagonal) $\rho$. Such operators naturally arise as reduced density matrices of many-particle quantum systems. In the case of fermions, the Pauli principle dictates a bound on the eigenvalues of $\gamma$, which in the simplest (spinless) case reads $\gamma\leq 1$. In this case, Lieb and Thirring \cite{LT1,LT2} proved a powerful lower bound on the kinetic energy $\Tr (-\Delta)\gamma$, where $\Delta$ is the Laplacian on $\R^d$, and the trace should really be interpreted as the one of the positive operator $-\nabla \gamma \nabla$. This bound is one of the key ingredients in their elegant proof of the stability of matter, first proved by Dyson and Lenard in \cite{DL}. It can be interpreted as a many-body uncertainly principle, and reads 
\begin{equation}\label{LTo}
\Tr (-\Delta)\gamma \geq C_d^{\rm LT} \int_{\R^d} \rho^{1+2/d}
\end{equation}
for some universal constant $C_d^{\rm LT}$ depending only on the space dimension $d$. The optimal value of this constant is not known, and for $d\geq 3$ was conjectured by Lieb and Thirring to equal the semi-classical Thomas--Fermi value, $C_d^{\rm TF} = 4\pi \frac d{d+2}  \Gamma(1+d/2)^{2/d}$. We refer to \cite{Fd} for the currently best known lower bounds, as well as to \cite{LTbook} for further information on Lieb--Thirring and related inequalities. We note that Lieb and Thirring proved \eqref{LTo} by first proving a dual inequality on the sum of the negative eigenvalues of Schr\"odinger operators, but direct proofs of \eqref{LTo} have since also been derived  \cite{Rumin,Lundholm,Fd}.

In \cite{Nam} Nam proved a Lieb--Thirring inequality with constant arbitrarily close to $C_d^{\rm TF}$, at the expense of a gradient correction term. In this paper we present an improved version of Nam's inequality, with a much simpler proof. Our proof is inspired by \cite[Thm.~3]{LLS}, where an analogous upper bound is proved (on the kinetic energy density functional, i.e., the infimum of $\Tr(-\Delta)\gamma$ for given $\rho$). 
Interestingly, the method can also be used for a lower bound, in a similar spirit as the method of coherent states, which can also be applied to give bounds in both directions \cite{LLa}, but seems to be more useful for the study of the dual problem, however.

Our main result is the following.

\begin{theorem}\label{thm:main}
Let $\eta:\R_+\to \R$ be a function with 
\begin{equation}\label{norm}
\int_0^\infty \eta(t)^2 \frac {dt}t = 1 = \int_0^\infty \eta(t)^2 t\, {dt}  
\end{equation}
and let $C_d^{\rm TF} = 4\pi \frac d{d+2}  \Gamma(1+d/2)^{2/d}$. 
For any trace-class $0\leq \gamma\leq 1$ on $L^2(\R^d)$ with density $\rho$, 
\begin{equation}\label{main:eq}
\Tr (-\Delta) \gamma \geq \frac{C_d^{\rm TF}}{\left(  \int_0^\infty \eta(t)^2 t^{d+1} dt\right)^{2/d}}  \int_{\R^d} \rho^{1+2/d} 
- \frac 4{d^2} \int_{\R^d}  |\nabla  \sqrt\rho|^2 \int_0^\infty\eta'(t)^2 t \,dt
\end{equation}
\end{theorem}

We note that under the normalization conditions \eqref{norm} we have $\int_0^\infty \eta(t)^2 t^{d+1} dt > 1$ by Jensen's inequality. In order for this integral to be close to $1$, $\eta^2$ needs to be close to a $\delta$-distribution at $1$, in which case the final factor in \eqref{main:eq} necessarily becomes large, however. 
A possible concrete choice is 
\begin{equation}\label{trial}
\eta(t) = (\pi \epsilon)^{-1/4}\exp\left( -  (\epsilon/2+ \ln t)^2/(2\epsilon)\right)
\end{equation}
for $\epsilon>0$. 
Then $\int_0^\infty \eta'(t)^2 t\,{dt} = (2\epsilon)^{-1}$ 
and
\begin{align*}
&\int_0^\infty \eta(t)^2 t^{1+x}  {dt}
=\exp\left( \epsilon  x(2+x)/4 \right)
\end{align*}
for any $x\in \R$. 
For this choice of $\eta$ the bound \eqref{main:eq} thus reads
$$
\Tr (-\Delta) \gamma \geq C_d^{\rm TF}  e^{-\epsilon(1 + d/2)  }  \int_{\R^d} \rho^{1+2/d} - \frac 2{d^2 \epsilon} \int_{\R^d}  |\nabla  \sqrt\rho|^2 
$$
for any $\epsilon>0$. A similar bound was proved by Nam in \cite{Nam}, but with the exponent $-1$ of $\epsilon$ in the gradient term replaced by $-3 - 4/d$. We don't expect the exponent $-1$ to be optimal, however. In fact, according to the Lieb--Thirring conjecture no correction term to the semiclassical expression should be needed at all for $d\geq 3$. Some correction term is needed for $d\leq 2$, but possibly the divergence of the prefactor as $\epsilon\to 0$ could be slower than in our bound.

As already pointed out in \cite{Nam}, one can combine an inequality of the form \eqref{main:eq} with the Hoffmann-Ostenhof inequality \cite{H2O}
\begin{equation}\label{h2o}
\Tr (-\Delta) \gamma \geq  \int_{\R^d}  |\nabla  \sqrt\rho|^2 
\end{equation}
to obtain a Lieb--Thirring inequality without gradient correction. The following is an immediate consequence of \eqref{main:eq} and \eqref{h2o}.

\begin{corollary}
For any trace-class $0\leq \gamma\leq 1$ on $L^2(\R^d)$ with density $\rho$, we have
\begin{equation}\label{LT}
\Tr (-\Delta) \gamma \geq   C_d^{\rm TF}  R_d \int_{\R^d} \rho^{1+2/d} 
\end{equation}
with 
\begin{equation}\label{def:Rd}
R_d =  \sup_{\eta}  \frac{1}{\left(  \int \eta(t)^2 t^{d+1} dt\right)^{2/d}} \frac 1 { 1+ \frac 4{d^2} \int \eta'(t)^2 t \,dt}
\end{equation}
where the supremum is over functions $\eta$ satisfying the normalization conditions \eqref{norm}. 
\end{corollary}

We shall show below that for $d\leq 2$, $R_d$ can be calculated explicitly. In fact, $R_1 = (-3/a)^{3}/2^4 \approx 0.132$, where $a \approx -2.338$ is the largest real zero of the Airy function, and $R_2 = 1/4$. We were not able to compute $R_d$ for $d\geq 3$, but it can easily be obtained numerically. For $d=3$, we find $R_d \approx 0.331$. In all these cases, our result is weaker than the best known one in \cite{Fd}, however, and also weaker than the one obtained in  \cite{Rumin} where \eqref{LT} was proved with $R_d = d/(d+4)$. 

\begin{proof}[Proof of Theorem~\ref{thm:main}]
The starting point is the following IMS type formula for any positive function $f:\R^d\to \R_+$, 
$$
\Delta = \int_0^\infty \eta(t/f(x)) \Delta \eta(t/f(x)) \frac {dt} t  + \frac{|\nabla f(x)|^2}{f(x)^2} \int_0^\infty \eta'(t)^2 t \,dt
$$
where we used the first normalization condition in \eqref{norm}. 
This follows from
$$\frac12\theta^2\Delta+\frac12\Delta \theta^2=\theta\Delta\theta +(\nabla\theta)^2$$
applied to $\theta(x)=\eta(t/f(x))$.
As a consequence, we have 
$$
\Tr (-\Delta) \gamma = -\int_{\R^d} \rho \frac{|\nabla f|^2}{f^2} \int_0^\infty \eta'(t)^2 t \,dt + \int_{\R^d} \int_0^\infty  p^2  \langle  \psi_{p,t}  | \gamma | \psi_{p,t} \rangle  \frac{dt}{t} dp
$$
where  $\psi_{p,t}(x) = (2\pi)^{-d/2} e^{ipx} \eta(t/f(x))$. Note also that 
$$
\int_{\R^d} \int_0^\infty  t  \langle  \psi_{p,t} | \gamma | \psi_{p,t} \rangle  {dt}\, dp = \int_{\R^d} \rho f^2 \int_0^\infty \eta(t)^2 t\, dt = \int_{\R^d} \rho f^2
$$
where we used the second normalization condition in \eqref{norm}. 
Hence
\begin{align*}
\Tr (-\Delta) \gamma &= -\int_{\R^d} \rho \frac{|\nabla f|^2}{f^2} \int_0^\infty \eta'(t)^2 t \,dt + \int \rho f^2 \\ & \quad + \int_{\R^d} \int_0^\infty  (p^2 - t^2)  \langle  \psi_{p,t} | \gamma | \psi_{p,t} \rangle  \frac{dt}{t} dp
\end{align*}
Since $0\leq \gamma\leq 1$ by assumption, we can get a lower bound on the last term as
$$
 \int_{\R^d} \int_0^\infty  (p^2 - t^2)  \langle \psi_{p,t} | \gamma | \psi_{p,t} \rangle  \frac{dt}{t} dp \geq  \int_{\R^d} \int_0^\infty  (p^2 - t^2)_-  \|  \psi_{p,t} \|^2  \frac{dt}{t} dp
$$
where $(\, \cdot\,)_- = \min\{ 0, \, \cdot\,\}$ denotes the negative part. 
Since
$$
 \|  \psi_{p,t} \|^2 =  \frac{1}{(2\pi)^{d}} \int_{\R^d} \eta(t/f(x))^2 dx
$$
we have
$$
\int_{\R^d} \int_0^\infty  (p^2 - t^2)_-  \|  \psi_{p,t} \|^2  \frac{dt}{t} dp = - \frac 1{ (2\pi)^{d}} \int_{|p|\leq 1} (1- p^2) dp \int_{\R^d} f^{d+2}  \int_0^\infty \eta(t)^2 t^{d+1} dt
$$
Altogether, we have thus shown that
\begin{align*}
\Tr (-\Delta) \gamma &\geq-\int_{\R^d} \rho \frac{|\nabla f|^2}{f^2} \int_0^\infty \eta'(t)^2 t \,dt + \int_{\R^d} \rho f^2 
\\ & \quad - \frac 1{ (2\pi)^{d}} \int_{|p|\leq 1} (1- p^2) dp \int_{\R^d} f^{d+2}  \int_0^\infty \eta(t)^2 t^{d+1} dt
\end{align*}
We now choose $f = c \rho^{1/d}$ and optimize over $c>0$. This gives \eqref{main:eq}.
\end{proof}

Finally, we shall analyze the optimization problem in \eqref{def:Rd}.  
Let $e_d>0$ denote  the ground state energy of $-\partial_t^2 - t^{-1} \partial_t + d^2 / (4 t^2)  +  t^d$ on $L^2(\R_+, t\, dt)$ (or, equivalently, of $-\Delta + |x|^d$ on $L^2(\R^{d+2})$). We claim that
\begin{equation}\label{claim}
R_d = \frac d 2 \left( \frac{ d+2}{2 e_d} \right)^{1+2/d}  
\end{equation}
To see this, let us note that by a straightforward scaling argument we can rewrite $R_d^{-1}$ as 
\begin{align}\nonumber
\frac 1{R_d}  &= \frac 4{d^2}  \inf_{\|\eta\|_2=1}  \left(  \int \eta(t)^2 t^{d+1} dt\right)^{2/d} \int \left(  \frac {d^2}{4t^2} \eta(t)^2  + \eta'(t)^2 \right)t \,dt
\\ & =  \frac 4{d^2}  \inf_{\|\eta\|_2=1} \inf_{\lambda>0}  \left( \frac 2{d \lambda} \right)^{2/d}  \left[ \frac d{d+2} \int \left(  \frac {d^2}{4t^2} \eta(t)^2 + \lambda t^d \eta(t)^2 + \eta'(t)^2 \right)t \,dt\right]^{1+2/d}
\end{align}
where $\|\eta\|_2$ denotes the $L^2(\R_+,t\, dt)$ norm, and we used the simple identity $a b^{x} = \frac{x^x}{(1+x)^{1+x}} \inf_{\lambda>0} \lambda^{-x} (a+\lambda b)^{1+x}$ for positive numbers $a$, $b$ and $x$. Taking first the infimum over $\eta$ for fixed $\lambda$ leads to the ground state energy of $-\partial_t^2 - t^{-1} \partial_t + d^2 / (4 t^2)  + \lambda t^d$, which a change of variables shows to be equal to $\lambda^{2/(d+2)}e_d$. Hence we arrive at \eqref{claim}.

For $d=1$, once readily checks that the ground state of $-\partial_t^2 - t^{-1} \partial_t + 1 / (4 t^2)  +  t$ equals $t^{-1/2} \mathrm {Ai}(t+a)$ with $a$ the largest real zero of the Airy function $\mathrm{Ai}$. In particular, $e_1 = -a$. For $d=2$ we find $e_2=4$ (the ground state energy of $-\Delta+|x|^2$ on ${\mathbb R}^4$),  
and the ground state of $-\partial_t^2 - t^{-1} \partial_t + 1 / t^2  +  t^2$ is given by $t e^{-t^2/2}$. 

One can also check that $R_d \to 1$ as $d\to \infty$. In fact, using \eqref{trial} as a trial state and optimizing over the choice of $\epsilon$, one finds 
$$
 R_d \geq \frac {\sqrt{ 1 +  \frac {2d^2}{1+d/2} } - 1}{\sqrt{ 1 +  \frac {2d^2}{1+d/2} } + 1}  \exp\left( - \frac {1+d/2}{d^2} \left(  \sqrt{ 1 +  \frac {2d^2}{1+d/2} } - 1\right) \right)= 1 - O(d^{-1/2})\,.
 $$

\bigskip
{\it Acknowledgments.} 
J.P.S. thanks the Institute of Science and Technology Austria for the hospitality and support during a visit where this work was done. J.P.S. was also partially supported by the VILLUM Centre of Excellence for the Mathematics of Quantum Theory (QMATH).

\end{document}